\documentclass[letterpaper]{scrartcl}

\usepackage{amssymb, amsthm, mathtools, cite}
\usepackage[hidelinks, citecolor=blue]{hyperref}
\hypersetup{colorlinks=true, urlcolor = blue, linkcolor=black}

\newtheorem{theorem}{Theorem}
\newtheorem{corollary}[theorem]{Corollary}

\newtheorem{proposition}[theorem]{Proposition}
\numberwithin{equation}{section}
\numberwithin{theorem}{section}

\newcommand{\rr}{{\mathbb{R}}}

\newcommand{\cp}{{\mathbb{C}_+}}
\newcommand{\ee}{{\mathbb{E}\,}}

\newcommand{\pp}{{\mathbb{P}}}
\newcommand{\oh}{{\mathcal{O}}}

\newcommand{\im}{{\operatorname{Im}\,}}
\newcommand{\re}{{\operatorname{Re }\,}}

\newcommand{\tr}{{\operatorname{Tr}\,}}

\newcommand{\beq}[1]{\begin{equation} \label{#1}}
\newcommand{\eeq}{\end{equation}}
\newcommand{\zel}{{z_\ell}}

\begin{document}
\addtokomafont{author}{\raggedright}
\title{\raggedright The Localization Transition in the Ultrametric Ensemble}
\author{\hspace{-.075in}Per von Soosten and Simone Warzel}
\date{\vspace{-.2in}}
\maketitle
\minisec{Abstract}
We study the hierarchical analogue of power-law random band matrices, a symmetric ensemble of random matrices with independent entries whose variances decay exponentially in the metric induced by the tree topology on $\mathbb{N}$. We map out the entirety of the localization regime by proving the localization of eigenfunctions and Poisson statistics of the suitably scaled eigenvalues. Our results complement existing works on complete delocalization and random matrix universality, thereby proving the existence of a phase transition in this model.
\bigskip

\section{Introduction}
One-dimensional lattice Hamiltonians with random long-range hopping provide a useful and simplified testing ground for the Anderson metal-insulator transition in more complicated systems. Two prominent examples of such models are the random band matrices \cite{PhysRevLett.64.1851,PhysRevLett.67.2405}, whose entries $H_{ij}$ are zero outside some band $|i-j| \le W$, and the power-law random band matrices (PRBM) \cite{RevModPhys.80.1355,PhysRevE.54.3221}, whose entries $H_{ij}$ have variances decaying according to some power of the Euclidean distance $|i-j|$. Even for these models, the mathematically rigorous understanding is far from complete, although there has been some progress; see \cite{MR2726110,0036-0279-66-3-R02, bourgadearxiv,MR1942858,MR3085669, MR2525652} and references therein. This article is concerned with a further simplification, the ultrametric ensemble of Fyodorov, Ossipov and Rodriguez \cite{1742-5468-2009-12-L12001}, which is essentially obtained by replacing the Euclidean distance in the definition of the PRBM with the metric induced by the tree topology.

The index space of the ultrametric ensemble is $B_n = \{1,2,..., 2^n\}$ endowed with the metric
\[d(x,y) = \min  \left\{ r \geq 0 \, | \, \mbox{$x$ and $y$ lie in a common member of $\mathcal{P}_r$} \right\},\]
where $\{\mathcal{P}_r\}$ is the nested sequence of partitions defined by
\[B_n = \{1, ..., 2^{r}\} \cup \{2^r + 1, ..., 2\cdot 2^r \} \cup ... \cup \{2^{n-r-1}2^r + 1, ..., 2^n\}.\]
The basic building blocks of the ultrametric ensemble are the matrices $\Phi_{n,r}: \ell^2(B_n) \to \ell^2(B_n)$ whose entries are independent (up to the symmetry constraint) centered real Gaussian random variables with variance
\beq{Phidef} \ee \left| \langle \delta_y, \Phi_{n,r} \delta_x \rangle \right|^2 = 2^{-r} \begin{cases} 2 & \mbox{ if } d(x,y) = 0\\ 1 & \mbox{ if } 1 \le d(x,y) \le r\\ 0 & \mbox{ otherwise. }\end{cases}
\eeq
Thus $\Phi_{n,r}$ is a direct sum of $2^{n-r}$ random matrices drawn independently from the Gaussian Orthogonal Ensemble (GOE) of size $2^r$. The ultrametric ensemble with parameter $c \in \rr$ refers to the random matrix
\beq{Hdef}H_n = \frac{1}{Z_{n,c}} \sum_{r =0}^n 2^{-\frac{(1+c)}{2}r} \Phi_{n,r}
\eeq
where $\Phi_{n,r}$ and $\Phi_{n,s}$ are independent for $r \neq s$. We choose the normalizing constant $Z_{n,c}$ such that
\[\sum_{y \in B_n} \ee \left| \langle \delta_y, H_n \delta_x \rangle \right|^2 = 1,\]
which means that $Z_{n,c}$ grows exponentially in $n$ in case $c < -1$ and $Z_{n,c}$ is asymptotically constant in case $ c > - 1 $. The original definition in \cite{1742-5468-2009-12-L12001} contains an additional parameter governing the relative strengths of the diagonal and off-diagonal disorder, but this parameter does not significantly alter our analysis and so we omit it altogether. Moreover, the authors of \cite{1742-5468-2009-12-L12001} constructed the block matrices $\Phi_{n,r}$ from the Gaussian Unitary Ensemble (GUE), and our results apply to both GOE and GUE blocks with only slight changes.

The ultrametric ensemble is a hierarchical analogue of the PRBM in a sense which was first introduced to statistical mechanical models by Dyson \cite{MR0436850}, and studied rigorously in the context of quantum hopping systems with only diagonal disorder in \cite{MR1063180, MR2352276, MR2413200, MR2909106,MR1463464,vonSoosten2017}. Thus, one expects the core features of the PRBM phase transition to be present in the ultrametric ensemble as well. Indeed, the authors of  \cite{1742-5468-2009-12-L12001} present arguments at a theoretical physics level of rigor as well as numerical evidence for a localization-delocalization transition in the eigenfunctions of $H_n$ as the parameter changes from $c > 0$ to $c < 0$. The effect of the Gaussian perturbations $\Phi_{n,r}$ on the spectrum of $H_n$ can be described dynamically by Dyson Brownian motion \cite{MR0148397} and, in this sense, the critical point $c=0$ is natural because it governs whether the evolution passes the local equilibration time of this system or not.

In this article we establish the localized phase by proving that the eigenfunctions remain localized and the level statistics converge to a Poisson point process if $c > 0$. Both statements are improvements over the rigorous results for band matrices, where the only localization result is due to Schenker \cite{MR2525652} and no proof of Poisson statistics is known. For the first main result in this paper, we recall that, by the Wegner estimate \cite{MR639135}, the infinite-volume density of states measure defined through
\beq{dosdef}\nu(f) = \lim_{n \to \infty} 2^{-n} \sum_{\lambda_j \in \sigma(H_n)} f(\lambda_j)
\eeq
has a bounded Radon-Nikodym derivative, the density of states, whose values we denote by $\nu(E)$.
\begin{theorem}[Poisson statistics] \label{poissonthm} Suppose $c > 0$ and let $E \in \rr$ be a Lebesgue point of $\nu$. Then, the random measure
\[\mu_n(f) = \sum_{\lambda \in \sigma(H_n)} f(2^n(\lambda - E))\]
converges in distribution to a Poisson point process with intensity $\nu(E)$ as $n \to \infty$.
\end{theorem}
The proof of Theorem \ref{poissonthm} is contained in Section \ref{poissonsection}.

The second main result of this paper says that if an eigenfunction of $H_n$ in some mesoscopic energy interval has any mass at some $x \in B_n$, then actually all but an exponentially small amount of the total mass is carried in an exponentially vanishing fraction of the volume near $x$ with high probability. We formulate this result in terms of the eigenfunction correlator
\[Q_n(x,y;W) = \sum_{\lambda \in \sigma(H_n) \cap W} |\psi_\lambda(x)\psi_\lambda(y)|,\]
which in completely delocalized regimes is typically given by
\[\sum_{\lambda \in \sigma(H_t) \cap W}  \sum_{y \neq x}|\psi_\lambda(x)\psi_\lambda(y)| \approx \sum_{\lambda \in \sigma(H_t) \cap W} 1 \approx 2^n|W|.\]
Thus, since $2^n|W|$ grows very large for mesoscopic spectral windows, it is a signature of localization if the eigenfunction correlator asymmptotically vanishes for small enough mesoscopic intervals $W$, as is proved in the following theorem.
\begin{theorem}[Eigenfunction localization]\label{localizationthm} Suppose $c > 0$ and let $E \in \rr$. Then, there exist $w, \mu, \kappa > 0$, $C<\infty$, and a sequence $m_n$ with $n - m_n \to \infty$ such that for every $x \in B_n$ the $ \ell^2$-normalized eigenfunctions satisfy
\[\pp\left( \sum_{y  \in B_n \setminus B_{m_n}(x)} Q_n(x,y;W) > 2^{-\mu n} \right) \le C \,2^{-\kappa n}\]
with
\[W = \left[E_0 - 2^{-(1-w)n}, E_0 + 2^{-(1-w)n}\right].\]
\end{theorem}
The proof of Theorem \ref{localizationthm} may be found in Section \ref{localizationsection}.

Our results gain additional interest upon noting that, when $c < -1$, the ultrametric ensemble has an essential mean field character and techniques originally developed for Wigner matrices show that the energy levels agree asymptotically with those of the GOE and that the eigenfunctions are completely delocalized. We will now roughly sketch how to apply these results in the present situation and state the corresponding theorems. The key observation is that the normalizing factor $Z_{n,c}$, which scales the spectrum to $\oh(1)$, is given by
\[Z_{n,c}^2 =  \sum_{y \in B_n} \ee \left| \langle \delta_y, \left(\sum_{r =0}^n 2^{-\frac{(1+c)}{2}r} \Phi_{n,r}\right) \delta_x\rangle \right|^2 =  \left(1 - 2^{-(1+c)(n+1)} \right)\,  \frac{1 + \oh(1)}{1-2^{-(1+c)}},\]
so that the spread
\[M_n \vcentcolon= \left(\max_{x,y \in B_n} \, \ee |\langle \delta_y, H_n\delta_x \rangle|^2 \right)^{-1} = \left\{  \begin{array}{lr}  Z_{n,c}^2 \, 2^{- o(n)}    & \mbox{if} \;  c \geq -2,   \\[1ex]
2^{ (1+o(1))n} & \mbox{if} \;  c < -2, 
\end{array}
\right. \]
grows like a positive power of the system size $2^n$ when $c < -1$. The results of \cite{MR3068390} then show that the semicircle law (i.e.\ $ \nu(E) = \sqrt{(4-E^2)_+}/(2\pi) $) is valid on scales of order $M_n^{-1}$ even for the matrices
\[\widetilde{H}_n =  \frac{1}{Z_{n,c}} \sum_{r =0}^{n-1} 2^{-\frac{(1+c)}{2}r} \Phi_{n,r} + \frac{1-\sqrt{T_n}}{Z_{n,c}} 2^{-\frac{(1+c)}{2}n} \Phi_{n,n} \]
with a small part of the final $\oh(1)$ Gaussian component removed. We set $ T_n = M_n^{-1+\delta} $ with $ \delta \in (0,1) $. The validity of the local semicircle law already implies the complete delocalization of the eigenfunctions in mesoscopic windows in the bulk of the spectrum (see~\cite[Thm. 2.21]{0036-0279-66-3-R02}).
\begin{theorem}[cf. \cite{0036-0279-66-3-R02, MR3068390}] For any compact interval $ I \subset (-2,2) $ there exist $\kappa, \epsilon > 0$ such that for all $ E \in I $ the $ \ell^2$-normalized eigenfunctions of $H_n$ in $[E-M_n^{-1}, E+M_n^{-1}]$ satisfy
\[\|\psi_{\lambda}\|_\infty =\oh(M_n^{-1/2 + \epsilon})\]
with probability $1 - \oh(N^{-\kappa})$.
\end{theorem}

Random matrix universality of the local statistics may be expressed by saying that the $k$-point correlation functions
\[\rho^{(k)}_{H_n}(\lambda_1, .. \lambda_k) =  \int_{\rr^{2^n - k}} \! \rho_{H_n} (\lambda_1, ..., \lambda_{2^n}) \, d\lambda_{k+1} ... \, d\lambda_{2^n},\]
the $k$-th marginals of the symmetrized eigenvalue density $\rho_{H_n}$, locally agree with the corresponding objects for the GOE asymptotically. For this, we employ the work of Landon, Sosoe and Yau \cite[Thm.~2.2]{landonsosoeyauarxiv} concerning the universality of Gaussian perturbations for
\[H_n = \widetilde{H}_n +  \frac{\sqrt{T_n}}{Z_{n,c}} 2^{-\frac{(1+c)}{2}n} \Phi_{n,n}.\]
For the statement of the theorem, let
\begin{align*}\Psi_{n,E}^{(k)}(\alpha_1,...,\alpha_k) &= \rho^{(k)}_{H_n}\left( E + 2^{-n} \frac{\alpha_1}{\rho_{sc}(E)}, ..., E + 2^{-n} \frac{\alpha_k}{\rho_{sc}(E)} \right)\\
 &-  \rho^{(k)}_{GOE}\left( E + 2^{-n} \frac{\alpha_1}{\rho_{sc}(E)}, ..., E + 2^{-n} \frac{\alpha_k}{\rho_{sc}(E)} \right),
\end{align*}
where $\rho^{(k)}_{GOE}$ is the $k$-point correlation function of the $2^n \times 2^n$ GOE and $\rho_{sc}$ is the density of the semicircle law.
\begin{theorem}[cf.~\cite{landonsosoeyauarxiv,MR3068390}] Suppose $c < -1$, $E \in (-2,2)$ and $k \geq 1$. Then,
\[\lim_{n \to \infty} \int_{\rr^k} \! O(\alpha) \Psi_{n,E}^{(k)}(\alpha) \, d\alpha = 0\]
for every $O \in C_c^\infty(\rr^k)$.
\end{theorem}

Summing up, these results rigorously prove the existence of a metal-insulator transition in the ensemble of ultrametric random matrices. In particular, our results allow an approach all the way to the critical point from the localized side $ c > 0 $, which greatly improves upon the best known corresponding result for random band matrices \cite{MR2525652}. However, the above arguments do not cover the regime $ c \in [-1,0) $, in which the local eigenvalue statistics are still expected to be of Wigner-Dyson-Mehta type as in the case $ c < -1 $ \cite{1742-5468-2009-12-L12001}.

\section{Proof of Poisson Statistics}\label{poissonsection}
When $ c > 0 $, the limit $ \lim_{n\to \infty} Z_{n,c}  \in (0,\infty)  $ exists. Thus, we may drop the normalizing constant $ Z_{n,c} $ from the definition of  $ H_n $ in the remainder of this paper without any loss of generality. The goal of this section is to prove Theorem \ref{poissonthm} by approximating $H_n  \equiv  \sum_{r=0}^n 2^{-\frac{1 + c}{2} r}\Phi_{n,r} $ with the truncated Hamiltonian
\beq{hnmdefinition} H_{n,m} =  \sum_{r=0}^m 2^{-\frac{1 + c}{2} r}\Phi_{n,r},
\eeq
which has the property that, for any $m \le k \le n$,
\beq{hnmdecomposition} H_{n,m} = \bigoplus_{j=1}^{2^{n-k}} H_{k,m}^{(j)},
\eeq
where each $H_{k,m}^{(j)}$ is an independent copy of $H_{k}$. Therefore
\[\mu_{n,m}(f) = \sum_{\lambda \in \sigma(H_{n,m})} f(2^n(\lambda - E)) \]
consists of $2^{n-m}$ independent components, a fact whose relevance to Theorem \ref{poissonthm} is contained in the following characterization of Poisson point processes.
\begin{proposition} \label{poissoncharacterization}
Let $\{\mu_{n,j} \, | \, j = 1, ..., N_n\}$ be a collection of point processes such that:
\begin{enumerate}
\item The point processes $\{\mu_{n,1}, ..., \mu_{n, N_n}\}$ are independent for all $n \geq 1$.
\item If $B \subset \rr$ is a bounded Borel set, then
\[\lim_{n \to \infty} \sup_{j \le N_n} \pp(\mu_{n,j}(B) \geq 1) = 0.\]
\item There exists some $c \geq 0$ such that if $B \subset \rr$ is a bounded Borel set with $|\partial B| = 0$, then
\[\lim_{n \to \infty} \sum_{j=1}^{N_n} \pp(\mu_{n,j}(B) \geq 1) = c|B|\]
and
\[\lim_{n \to \infty} \sum_{j=1}^{N_n} \pp(\mu_{n,j}(B) \geq 2) = 0 .\]
\end{enumerate}
Then, $\mu_n = \sum_j \mu_{n,j}$ converges in distribution to a Poisson point process with intensity $c$.
\end{proposition}
We recall \cite{MR3364516} that a sequence of point processes $\mu_n$ converges in distribution to $\mu$ whenever
\[\lim_{n \to \infty} \ee e^{-\mu_n(P_z)}= \ee e^{-\mu(P_z)}\]
for all $z \in \cp$, where $P_z$ is the rescaled Poisson kernel
\beq{poissonkerneldef} P_z(\lambda) = \im \frac{1}{\lambda - z} = \frac{\im z}{(\lambda - \re z)^2 + (\im z)^2}.
\eeq
Hence, Theorem \ref{poissonthm} follows by furnishing a sequence $m_n$ such that Proposition~\ref{poissoncharacterization} applies to $\mu_{n, m_n}$ and
\beq{needtoprove} \lim_{n \to \infty} \ee e^{-\mu_{n,m_n}(P_z)} = \lim_{n \to \infty} \ee e^{-\mu_n(P_z)}
\eeq
for all $z \in \cp$.

The difference $ H_n - H_{n,n-1} =  \sqrt{t\,}  \Phi_{n,n} $ is a Gaussian perturbation with time parameter $t = 2^{-(1+c) n }$. Therefore, Theorem 1.1 of \cite{resflow} shows that there exists $ C_z < \infty $ such that for all $\ell \geq n$ we have
\begin{align} \label{resflowbound1} \frac{1}{2^n}\, \ee \left| \tr (H_n - \zel)^{-1} - \tr (H_{n,n-1} - \zel)^{-1}\right| &\le C_z \, 2^{-\frac{c}{2}n-1}\,\left( 1+ 2^{3(\ell-n)}\right) \nonumber \\
&\leq C_z \, 2^{-\frac{c}{2}n} \,  2^{3(\ell-n) } 
\end{align}
with $\zel = E + 2^{-\ell}z$. Our strategy in achieving \eqref{needtoprove} thus consists of applying  \eqref{resflowbound1} to the finite-volume density of states measures
\[\nu_n(f) = 2^{-n} \tr f \left(H_n\right) ,   \quad \nu_{n,m}(f) = 2^{-n} \tr f \left(H_{n,m}\right) \]
in an iterative fashion.

\begin{theorem} \label{approximationbynunm} There exist $C_z < \infty$ and $\delta > 0$ such that
 \[ \ee \left| \nu_{n}\left(P_\zel\right) - \nu_{n,m}\left(P_\zel\right) \right| \le C_z\, 2^{3(\ell - (1+\delta) m)}\]
 for all $\ell \geq n$.
\end{theorem}
\begin{proof} The estimate \eqref{resflowbound1} proves that
\beq{thmappl} \ee \left| \nu_{k}\left(P_\zel\right) - \nu_{k,k-1}\left(P_\zel\right) \right| \le C_z\,2^{3(\ell - (1+\delta) k)}
\eeq
with $\delta = c/6$ when $\ell \geq k$. Since $\nu_n - \nu_{n,m}$ is given by a telescopic sum,
\[\nu_n(P_\zel) - \nu_{n,m}(P_\zel) =  \sum_{k=m+1}^{n} \left( \nu_{n,k}(P_\zel) -  \nu_{n,k-1}(P_\zel)\right),\]
the decomposition~\eqref{hnmdecomposition} implies that
\beq{telescopicterm} \nu_{n,k}(P_\zel) -  \nu_{n,k-1}(P_\zel)= 2^{-(n-k)} \sum_{j=1}^{2^{n-k}} \left( \nu_{k}(P_\zel) -  \nu_{k,k-1}(P_\zel) \right).
\eeq
Applying \eqref{thmappl} to each term in \eqref{telescopicterm} yields
\begin{align*} \ee \left| \nu_{n}\left(P_\zel\right) - \nu_{n,m}\left(P_\zel\right) \right| & \le \sum_{k=m+1}^{n} C_z\, 2^{3(\ell - (1 + \delta) k)} \le C_z\, 2^{3(\ell - (1 + \delta) m)}. 
\end{align*}
\end{proof}

Theorem \ref{approximationbynunm} has two important implications for the measures $\mu_n$ and $\mu_{n,m}$ which are based on the identities $\mu_n(P_z) =  \nu_n\left(P_{z_n}\right)$ and $\mu_{n,m}(P_z) =  \nu_{n,m}\left(P_{z_n}\right)$. The first of these enables us to find a suitable sequence $\mu_{n,m_n}$ satisfying \eqref{needtoprove}.
\begin{corollary}\label{divisibility}
There exists a sequence $m_n$ with $m_n \to \infty$ and $n - m_n \to \infty$ such that
\[\lim_{n \to \infty} \ee \left|\mu_n(P_z) -\mu_{n,m_n}(P_z) \right| = 0\]
for all $z \in \cp$.
\end{corollary}
\begin{proof} Since $\delta > 0$, there exists a sequence $m_n$ with $m_n \to \infty$, $n - m_n \to \infty$ and $n - (1 + \delta)m_n \to -\infty$. By applying Theorem~\ref{approximationbynunm} with $\ell = n$, we obtain
\[ \ee \left|\mu_n(P_z) - \mu_{n,m_n}(P_z)\right| \le C_z 2^{3(n - (1 + \delta)m_n)} \to 0.\]
\end{proof}

To finish the proof of Theorem \ref{poissonthm}, we need to show that $\mu_{n,m_n}$ satisfies the hypothesis of Proposition~\ref{poissoncharacterization}. By~\eqref{hnmdecomposition}, $\mu_{n,m_n}$ is a sum of point processes
\[\mu_{n,m_n} = \sum_{j=1}^{2^{n-m_n}} \mu_{m_n,j}\]
with independent $\mu_{m_n, j}$. If $B \subset \rr$ is a bounded Borel set, the theorem of Combes-Germinet-Klein \cite{MR2505733} asserts that 
$\pp(\tr 1_B(H_{m}) \geq \ell) \le \left(C \, 2^m |B| \right)^\ell/ \ell! $
and hence for any $ \ell \geq 0 $:
\begin{equation}\label{cgkbound}
\pp(\mu_{m_n, j}(B) \geq \ell) \le \frac{(C |B| \, 2^{m_n-n})^\ell}{\ell!}.
\end{equation}
Since $n-m_n \to \infty$, the first hypothesis of Proposition~\ref{poissoncharacterization} follows. Writing
\[X(n, \ell) = \sum_{j=1}^{2^{n-m}} \pp(\mu_{m_n, j}(B) \geq \ell),\]
~\eqref{cgkbound} implies
\[X(n,\ell) \le 2^{n-m_n} \frac{(C |B|\, 2^{m_n-n})^\ell}{\ell!} \to 0\]
when $\ell \geq 2$. In particular, $X(n, 2) \to 0$ and the last hypothesis of Proposition~\ref{poissoncharacterization} is satisfied. It remains to prove the remaining hypothesis of Proposition~\ref{poissoncharacterization}, which is the second important consequence of Theorem \ref{approximationbynunm} and is contained in the following theorem.
\begin{theorem}  Let $B\subset \mathbb{R} $ be a bounded Borel set. Then,
\[\lim_{n \to \infty} X(n,1) = \nu(E)|B|.\]
\end{theorem}
\begin{proof} By \eqref{hnmdecomposition} we have $ \ee \nu_{p,n} = \ee \nu_{n} $ for any $ p \geq n$, and so we conclude from Theorem~\ref{approximationbynunm} with $\ell = n$ that
\begin{align*}\lim_{n \to \infty} \left|\ee \left[\nu_n(P_{z_n}) - \nu(P_{z_n})\right]\right| &= \lim_{n \to \infty}\lim_{p \to \infty} \left|\ee \left[ \nu_n(P_{z_n}) - \nu_p(P_{z_n}) \right]\right|\\
&= \lim_{n \to \infty}\lim_{p \to \infty} \left| \ee \left[ \nu_{p,n}(P_{z_n}) - \nu_p(P_{z_n}) \right]\right|\\
&\le  \lim_{n \to \infty} C_z \, 2^{-3\delta n} = 0.
\end{align*}
This shows that the measures $\lambda_n(B) = 2^n \nu(2^{-n}B + E)$ satisfy
\[\lim_{n \to \infty}\left( \ee \mu_n(P_z) - \lambda_n(P_z)\right)= 0\]
and Corollary \ref{divisibility} implies that also
\begin{equation}\label{intensityconvergence}
\lim_{n \to \infty}\left(\ee \mu_{m_n}(P_z) - \lambda_n(P_z)\right)= 0.
\end{equation}
For any bounded Borel set $B \subset \rr$, the indicator $1_B$ is in the $L^1$-closure of the finite linear combinations from the set $\{P_z \, | \, z \in \cp\}$ and the measures $\ee \mu_n$ are absolutely continuous with uniformly bounded densities by the Wegner estimate. Together, these two observations yield that \eqref{intensityconvergence} is valid for any bounded Borel set $B \subset \rr$. Moreover, since $E$ is a Lebesgue point of $\nu$,
\[\lim_{n \to \infty} \lambda_n(B) = \lim_{n \to \infty} 2^n \nu(2^{-n}B + E) = \nu(E)|B|,\]
and hence we have shown that
\beq{indicatorconvergence} \lim_{n \to \infty} \ee \mu_{m_n}(B) = \nu(E)|B|.
\eeq
Since $\mu_{n_m, j}(B)$ takes values in the non-negative integers
\[\lim_{n \to \infty} X(n,1) = \lim_{n \to \infty} \sum_{j=1}^{2^{n - m}} \ee \mu_{n_m, j}(B) - \lim_{n \to \infty} \sum_{\ell \geq 2} X(n,\ell)\]
so ~\eqref{cgkbound}, \eqref{indicatorconvergence} and the dominated convergence theorem give
\[\lim_{n \to \infty} X(n,1) = \lim_{n \to \infty} \sum_{j=1}^{2^{n - m}} \ee \mu_{n_m, j}(B) = \nu(E)|B|.\]
\end{proof}

\section{Proof of Eigenfunction Localization}\label{localizationsection}
In this section, we prove Theorem \ref{localizationthm} by comparing the eigenfunctions of $H_n$ with the obviously localized eigenfunctions of $H_{n,m}$. As in Section~\ref{poissonsection}, we drop the normalizing constant $ Z_{n,c} $ from the definition of $ H_n$. The core of this argument again consists of resolvent bounds for Gaussian perturbations, and thus we consider the Green functions
\[G_n(x,y;z) = \langle \delta_y, (H_n - z)^{-1} \delta_x \rangle , \quad G_{n,m}(x,y;z) =  \langle \delta_y, (H_{n,m} - z)^{-1} \delta_x \rangle.\]
If $\eta = 2^{-(1 + \ell)n}$ for some $\ell > 0$, Theorem 1.2 of \cite{resflow} proves that there exists $ C < \infty $ such that
\begin{align*} 2^{-k}\sum_{y \in B_k(x)} \ee \left| G_k\left(x,y; E+i\eta\right)-  G_{k,k-1}\left(x,y;E+i\eta\right) \right| &\le C \, 2^{-\frac{c}{2}k} \left(1 + 2^{3( (1+\ell) n - k)}\right)\\
&= C\,2^{3(1 + \ell)n-3(1+\delta)k}
\end{align*}
with $\delta = c/6$ whenever $k \le n$. Iterating this result, we see that
\begin{align*} & 2^{-n} \sum_{y \in B_n} \ee \left| G_n\left(x,y; E+i\eta\right)-  G_{n,m}\left(x,y;E+i\eta \right) \right|\nonumber\\
&\le 2^{-n} \sum_{k=m+1}^n\sum_{y \in B_n} \ee \left| G_{n,k}\left(x,y; E+i\eta\right)-  G_{n,k-1}\left(x,y;E+i\eta \right) \right|\nonumber\\
&= 2^{-n} \sum_{k=m+1}^n \sum_{y \in B_k(x)} \ee \left| G_{k}\left(x,y; E+i\eta\right)-  G_{k,k-1}\left(x,y;E+i\eta \right) \right|\nonumber\\
&\le 2^{-n} \sum_{k=m+1}^n 2^k C\,2^{3(1 + \ell)n-3(1+\delta)k} \le C\,2^{(3(1+\ell) - 1) n} \, 2^{-(3(1+\delta)-1)m}.
\end{align*}
Since $\delta > 0$, we can choose $\ell > 0$, $\epsilon \in (0,1)$, and $w \in (0, \ell)$ such that
\[2\mu \vcentcolon= (1 - \epsilon)(3(1 + \delta)-1) - (3(1+\ell) -1)  - w > 0.\]
Thus, setting $m_n = (1 - \epsilon)n$ and
\[W = \left[E - 2^{-(1-w)n}, E + 2^{-(1-w)n}\right],\]
and using that $G_{n,m}(x,y;z) = 0$ if $y \notin B_m(x)$ show that
\[\sum_{y \in B_n \setminus B_{m_n}(x)} \ee \int_W \!  \left| \im G_n\left(x,y; E+i\eta\right) \right| \, dE \le C\,2^{-2\mu n}.\]
Applying Markov's inequality, we arrive at
\[\pp\left(\sum_{y \in B_n \setminus B_{m_n}(x)} \int_W \!  \left| \im G_n\left(x,y; E+i\eta\right) \right| \, dE > 2^{-\mu n}\right) \le C\, 2^{-\mu n},\]
so Theorem \ref{localizationthm} follows from Theorem 5.1 of \cite{resflow}, which says that
\[ \sum_{y \in B_n \setminus B_{m_n}(x)} Q_n(x,y;W) \le \sum_{y \in B_n \setminus B_{m_n}(x)} \int_W \! \left| \im G_n(x,y;E+i\eta \right| \, dE + \frac{\log 2^n}{2^{wn}}\]
with probability $1 - \oh\left(2^{(w-\ell)n}\right)$.

\subsection*{Acknowledgments}
We thank Y. V. Fyodorov for introducing us to the ultrametric ensemble. This work was supported by the DFG (WA 1699/2-1).
\bibliographystyle{abbrv}
\bibliography{UltrametricEnsemble}
\bigskip
\begin{minipage}{0.5\linewidth}
\noindent Per von Soosten\\
Zentrum Mathematik, TU M\"{u}nchen\\
Boltzmannstra{\ss}e 3, 85747 Garching, Germany\\
\verb+vonsoost@ma.tum.de+
\end{minipage}%
\begin{minipage}{0.5\linewidth}
\noindent Simone Warzel\\
Zentrum Mathematik, TU M\"{u}nchen\\
Boltzmannstra{\ss}e 3, 85747 Garching, Germany\\
\verb+warzel@ma.tum.de+
\end{minipage}
\end{document}